\documentclass[a4paper, oneside, 12pt]{article}
\usepackage{amsmath,amsthm,amssymb}
\usepackage{ bbold }
\usepackage[utf8]{inputenc}
\usepackage[T1]{fontenc}
\usepackage{graphicx}
\usepackage{longtable}
\usepackage[hidelinks]{hyperref}
\usepackage{cleveref}
\usepackage{tikz} 
\usepackage{pgfplots}
\usetikzlibrary{arrows}
\pgfplotsset{compat=1.15}
\usepackage{mathrsfs}
\usepackage{physics}
\usepackage{caption,subcaption}
\usepackage{comment}
\usepackage{authblk}
\hypersetup{colorlinks,linkcolor={blue},citecolor={blue},urlcolor={red}}  

\usepackage[left=2.5cm, right=2.5cm, top=2.5cm, bottom=2.5cm, bindingoffset=1.5cm, head=15pt]{geometry} 
\usepackage{setspace}
\usepackage{fancyhdr}
\usepackage{qcircuit}
\usepackage{algorithm,algpseudocode}
\usepackage{enumitem}
\usepackage{orcidlink}
\newtheoremstyle{break}
  {\topsep}{\topsep}%
  {\itshape}{}%
  {\bfseries}{}%
  {\newline}{}%
\theoremstyle{break}
\newtheorem{theorem}{Theorem}[section]
\newtheorem{corollary}{Corollary}[theorem]
\newtheorem{lemma}[theorem]{Lemma}
\newtheorem{conjecture}[theorem]{Conjecture}
\newtheorem{remark}[theorem]{Remark}

\newtheorem{definition}{Definition}[section]

\newcommand{\bs}[1]{\boldsymbol{#1}}

\definecolor{ududff}{rgb}{0.30196078431372547,0.30196078431372547,1}

\title{The Overlap Gap Property limits limit swapping in the QAOA}
\author[1,2]{\orcidlink{0009-0005-5505-2590}Mark Goh }
\affil[1]{Institute of Material Physics in Space, German Aerospace Center}
\affil[2]{Institute for Theoretical Physics, University of Cologne}
\date{\today}

\begin{document}

\maketitle

\begin{abstract}
The Quantum Approximate Optimization Algorithm (QAOA) is a quantum algorithm designed for Combinatorial Optimization Problem (COP). We show that if a local algorithm is limited in performance at logarithmic depth for a spin glass type COP with an underlying Erd\"os--R\'enyi hypergraph, then a random regular hypergraph is similarly limited in performance as well. As such, we re-derived the fact that the average-case value obtained by the QAOA for even $q\ge 4$, Max-$q$-XORSAT is bounded away from optimality when optimised using asymptotic analysis due to the Overlap Gap Property (OGP). While this result was proven before, the proof is rather technical compared to ours. In addition, we show that the earlier result implicitly also implies limitation at logarithmic depth $p \le \epsilon \log n$ providing an improvement over limitation at constant depth. Furthermore, the extension to logarithmic depth leads to a tightening of the upper bound that the QAOA outputs at logarithmic depth for MaxCUT and Max-$q$-XORSAT problems.
We also provide some numerical evidence the limitation should be extended to odd $q$ by showing that the OGP exists for the Max-$3$-XORSAT on random regular graphs.
\end{abstract}

\newpage 

\tableofcontents

\section{Introduction}
Combinatorial Optimization Problems (COPs) are notoriously difficult even as a decision problem \cite{Karp1972} --- well known examples include the travelling salesman problem \cite{steele1997probability} and finding the ground state of a spin glass Hamiltonian \cite{Gamarnik_2022}.  Rather than attempting to find an exact solution, one is often interested in approximate solutions. One such algorithm is the Quantum Approximate Optimization Algorithm (QAOA) introduced by Farhi et al.\ \cite{QAOA}, a type of variational quantum algorithm that, given $p$ layers, uses $2p$ optimisation parameters.\\

Attempting to evaluate the expectation value of the QAOA is incredibly difficult. Naively, given a problem with size $n$, each parameter in the QAOA requires a sum over $2^n$ terms. In a series of works starting with \cite{SKQAOAFarhi2022quantumapproximate}, algorithms to evaluate the expectation value of the QAOA on $q$-spin glass models with time complexity independent of $n$ have been found with increasing performance. The best known one for evaluating $q$-spin glass is found in \cite{QAOAspinglass} with a time complexity of $\mathcal{O}(p^2 4^p)$ using algebraic techniques.\\

Another line of research with respect to the QAOA is to prove its limitation in performance at shallow depths via the Overlap Gap Property (OGP). One of the first applications to show the limitation of performance uses locality properties of the QAOA. Thus, at shallow depth, the QAOA does not explore the whole graph underlying a COP and is unable to output a solution that beats the OGP barrier for the Maximum Independent Set problem \cite{QAOA_seegraph}. The limitation of the QAOA as a result of the OGP has been predominately limited on sparse graphs but a breakthrough came in \cite{Basso_2022} using a dense-from-sparse relation between complete graphs and sparse graphs to show that the QAOA is also limited in performance even if it sees the whole graph. Furthermore, they showed that for any constant $p$ in the asymptotic analysis, the QAOA is unable to surpass the OGP barrier. Assuming  a stronger version of the OGP, a similar and slightly better result shows limitation at super constant depth, $p\sim \mathcal{O}(\log \log n)$, for dense graphs \cite{Anshu2023concentrationbounds}.\\

In this paper, we re-derive a result of \cite{Basso_2022} that for the Max-$q$-XORSAT problem, and equivalently the mean-field $q$-spin glass, the QAOA is unable to find the optimal value even if $p$ goes to infinity for even $q\ge 4$ if we swap the order of limits, the thermodynamic limit and the run time of the algorithm (i.e.\ taking $\lim_{p\rightarrow \infty}\lim_{n\rightarrow \infty}$ rather than $\lim_{n\rightarrow \infty}\lim_{p\rightarrow \infty}$). More precisely, the analysis done by various authors is to study the performance of the QAOA asymptotically by taking the problem size $n$ to infinity and studying the output of the QAOA at constant depths $p$ results in the underlying graph explored by the QAOA to appear as a tree. The parameters that optimises the QAOA's performance in this instance are known as tree-parameters that work well on any Hamiltonian instance \cite{wybo2024missingpuzzlepiecesperformance,fixed-angle}. While this is a re-derivation of a known proof, their proof is rather technical compared to ours. Furthermore, our theorem provides some hints that their proof can be extended beyond constant $p$ and hold when $p \le \epsilon \log n$ for dense graphs since their method  is implicitly valid at logarithmic depth for sparse graphs.\\

The paper is organised as follows: In \cref{sec:Background} we give a brief background to random graphs, spin glass problems, the OGP, and the QAOA; in \cref{sec:results}, we summarise what is known in literature about the results of the QAOA on spin glasses and their equivalence between the mean-field and dilute spin model; in \cref{sec:Conjecture} we formalise the point about the OGP in random regular hypergraphs that was mentioned in \cite{QAOAspinglass} and re-derive the result in \cite{Basso_2022}, that the QAOA cannot find the optimal value for a spin glass COP even if the algorithm runs indefinitely under limit swapping. Following that, we outline a proof to affirm the theorem while providing some numerical evidence that the proof should extend to odd $q$-spin glass as well.

\subsection{Statement of result}

The main result of this work is to show that the OGP exists as a limitation for Max-$q$-XORSAT on a random regular hypergraph with sufficiently large degree. This is done via the following theorem

\begin{theorem}
    If the Overlap Gap Property limits the performance of a local algorithm on the Erd\"os--R\'enyi hypergraph at logarithmic depth $p$, it also limits the performance on a random regular hypergraph at logarithmic depth $p$.
\end{theorem}

The proof is done via contradiction. First, we show that we can trim the Erd\"os--R\'enyi hypergraph of average degree $\lambda$ to a regular hypertree removing at most $\mathcal{O}(1/\lambda^{\log \lambda})$ fraction of edges. Then, one can form a $\lambda$-regular hypergraph from the hypertree. Assuming an algorithm, that is limited by the OGP, is able to find a near optimal solution on the regular hypergraph leads to a contradiction since such an algorithm is unable to find near optimal solution on Erd\"os--R\'enyi hypergraph. Thus, the OGP also acts as a barrier to optimization for random regular graph.\\

This limitation at logarithmic depth local algorithms leads us to re-derive a recently discovered theorem and improve upon its results:
 
\begin{theorem}[Informal, theorem 3 and 4 of \cite{Basso_2022}]
    When $q\ge 4$ and is even, the  performance of the QAOA is limited at logarithmic depth for the Max-$q$-XORSAT with an underlying $D$-regular $q$-uniform hypergraphs. The performance is strictly upper-bounded by the OGP since
    \begin{align}
         \lim_{n\rightarrow\infty}
    \lim_{p\rightarrow \epsilon \log n}
    \frac{1}{|E|}\expval{ H_{XOR}^q}{\boldsymbol{\gamma},\boldsymbol{\beta}} =  \frac{1}{2}+
    \nu_{\infty}^{[q]}(D,\boldsymbol{\gamma},\boldsymbol{\beta}),
    \end{align}
    with $\nu_{\infty}^{[q]}(D,\boldsymbol{\gamma},\bs{\beta})<\nu_{\infty}^{[q]}(\boldsymbol{\gamma},\bs{\beta})<\Pi_q$.
\end{theorem}

There are several immediate corollaries of this result for the QAOA. The first of which was noted in \cite{QAOAspinglass} as a side-note.
\begin{corollary}
    Optimising the QAOA using the algorithm in \cite{QAOAspinglass,Basso_2022} only allows it to perform equal in performance to local classical algorithms thus providing no quantum advantage.
\end{corollary}

The above corollary is a result of optimising the QAOA under limit swapping of the algorithm run time $p$ and the problem size $n$. This therefore results in the following corollary:

\begin{corollary}
    If a COP exhibits the OGP, then optimising QAOA via limit swapping (i.e.\ using the tree parameters) results in sub-optimal performance.
\end{corollary}

\section{Background}
\label{sec:Background}

\subsection{Random Graphs}

Here we standardize the notation we use to denote a hypergraph. A hypergraph $G=(V,E)$ has $|V|=n$ vertices, and $|E|=m$ edges. A graph is $q$-uniform if every hyperedge connects to exactly $q$ vertices. A graph is $d$-regular if every vertex has degree $d$. Conventionally, an instance of the Erd\"os--R\'enyi--(Gilbert) $q$-uniform hypergraph $G=\mathbb{G}^q_{ER}(n,p)$ is a random graph with $n$ vertices where each hyperedge is added with probability $p$. The original Erd\"os--R\'enyi $q$-uniform hypergraph $G=\mathbb{G}^q_{ER}(n,m)$ is chosen randomly from the set of hypergraph with $n$ vertices and $m$ hyperedges. The former is now more frequently used. The two types of Erd\"os--R\'enyi graphs are similar to each other when $np=m$. In fact, it has been shown that the two types of random graphs are asymptotically equivalent under certain conditions \cite{luczak1990equivalence}.\\

Another type of random hypergraph of interest is the $d$-regular $q$-uniform hypergraph $G=\mathbb{R}^q(n,d)$ where we implicitly assume that $nd=qx$ for some integer $x$. Unlike Erd\"os--R\'enyi graphs that can be generated randomly, there is no easy unbiased way to generate such graphs, though one such method is known as the configuration model introduced by Bollob\'as \cite{BOLLOBAS1980311}.

\subsection{Spin glass}
\label{subsec:spinglass}
In this paper, we focus on the mean field spin glass model and the related dilute spin glass model. For the mean field model, the main goal, roughly speaking, is to find the ground state energy of a spin glass Hamiltonian. A widely studied model is the Sherrington--Kirkpatrick (SK) model \cite{SK_Model}. More generally, an Ising mean field $q$-spin model is given by the following
\begin{align}
   H_q(\bs{z})= \sqrt{\frac{q!}{2n^{(q-1)}}}\sum_{j<k<\dots <q} J_{jk\dots q} z_j z_k\dots z_q,
\end{align}
where the couplings are randomly chosen over a normal distribution and $z_i\in\{-1,1\}$.\\
    
\begin{align}
    \Pi_q = \lim_{n\rightarrow \infty} \max_{\bs{z}\in \Sigma_{n} }
    \frac{H_q(\bs{z})}{n}.
\end{align}

The dilute spin glass model is also known as the XOR-satisfiability (XORSAT) problem. Specifically, given a $q$-uniform hypergraph $G=(V,E)$ where $E\subset V^q$, and a given signed weight $J_{i_1,\dots,i_q}\in \{-1,+1\}$, Max-$q$-XORSAT is the problem of maximising the following cost function
\begin{align}
     H_{XOR}^q(\bs{z})= \sum_{(i_1,\dots,i_q)\in E} \frac{1}{2} (1+ J_{i_1 i_2 \dots i_q}z_{i_1} z_{i_2}\dots z_{i_q}).
\end{align}
In terms of the hypergraph, $z_{i_j}$ refers to the vertices and $J_{i_1 i_2 \dots i_q}$ refers to the hyperedge.\\

We say that an instance of the problem is satisfied if there is an assignment of values to the bit-string $\bs{z}$ which satisfies all the clauses (i.e.\ $ H_{XOR}^q(\bs{z})=|E|$). Otherwise, we say it is unsatisfiable. Suppose we fix $d>q$ sufficiently large so that we are in the unsatisfiable regime. The maximum number of satisfiable equations in an instance of random XORSAT for both $\mathbb{G}^q_{ER}(n,p)$, and $\mathbb{R}^q(n,d)$ has been found \cite{sen2017optimization} to be
\begin{align}
    \label{eq:asym_xorsat}
    \frac{1}{|E|}\max_{\bs{z}} H_{XOR}^q(\bs{z})
    = \frac{1}{2}+ \Pi_q \sqrt{\frac{q}{2d}} + \mathcal{O}(1/\sqrt{d})
\end{align}

\subsection{Overlap Gap Property}

One major obstacle to finding optimal solutions for COPs is known as the Overlap Gap Property (OGP). The term was introduced in \cite{OGPfirst}, though the concept was already used by various authors \cite{OGPsource,OGP_Source2}.\\

For the definition of the OGP, one can informally think that for certain choices of disorder $J$, there is a gap in the set of possible pairwise overlaps of near-optimal solution. Informally, for every two near optimal solution $\bs{z}^1,\bs{z}^2$, it is the case that the distance between them is either extremely small, or extremely large. Formally, we define the OGP for a single instance as the following:
\begin{definition}[Overlap Gap Property \cite{OGP}]

   For a general maximization problem with random input $J$, the OGP holds if there exists some $\epsilon >0,$ with $0\leq\mu_1 < \mu_2$ such that for every $\bs{z}^1,\bs{z}^2$ that is an $\epsilon$-optimal solution
    \begin{align}
        H_J(\bs{z}^i) \ge \max_{\bs{z}\in 2^n } H_J(\bs{z}) - \epsilon,
    \end{align}
    it holds that the (normalised) overlap between them is either less than $\mu_1$ or greater than $\mu_2$
    \begin{align}
    |R_{1,2}|\in [0,\mu_1] \cup [\mu_2,1],
    \end{align}
    where
    \begin{align}
        R_{1,2} = \frac{1}{n} \sum_{i=1}^n z_i^1 z_i^2.
    \end{align}
\end{definition}
Rather than using the overlap, it is often easier to visualize sets of solutions as clusters using the Hamming distance between two states via the relation
\begin{equation}
    \frac{H_{1,2}}{n} = \frac{1}{2}|1-R_{1,2}| \in [0,a] \cup [b,1]
\end{equation}
where $H_{1,2}$ is the Hamming distance between states $\bs{z}_1$ and $\bs{z}_2$ and $0<a<b<1$.\\

The first interval is trivial as we can simply choose the overlap $\bs{z}^1$ with itself. It is the existence of the second overlap, or rather the non-existence of overlap in the interval $(a, b)$, that is difficult to prove.\\

A general version of it is known as the ensemble-OGP introduced in \cite{Chen_2019} or coupled-OGP as used in \cite{chou2022limitations}. This version is required to prove limitations of local algorithm for technical reasons and requires an interpolation scheme between two different instances of Erd\"os--R\'eyni graphs. For spin glass, a branching OGP has been developed that makes use of the ultrametric structure in the Parisi solution \cite{huang2023algorithmic}.\\

Informally speaking, the OGP limits the performance of algorithms by first showing that the set of near optimal solutions exhibits a strong clustering property. That is, with high probability, a gap exists in their overlaps or equivalently, a gap in the hamming distance between near optimal solutions. Then, one proceeds to show that if an algorithm is able to find arbitrary near optimal solution, the algorithm outputs a solution that is in the region forbidden by the OGP. In the context of the QAOA, one typically uses the concentration of measure to prove the limitation of the QAOA at shallow depth \cite{Basso_2022,chou2022limitations}. We refer the reader to the review papers by Garmanik \cite{Gamarnik_2022,OGP} for details on how the OGP limits the performance of algorithms.

\subsection{QAOA}
The QAOA is a local quantum algorithm \cite{chou2022limitations} designed to find approximate solutions to combinatorial optimization problems \cite{QAOA}. The goal is to find a bit string $\bs{z}\in \{-1,+1\}^n$ that maximizes the cost function $C(\bs{z})$. Given a classical cost function $C$, we can define a corresponding quantum operator $\hat{C}$ that is diagonal in the computational basis, $\hat{C}\ket{\bs{z}}=C(\bs{z})\ket{\bs{z}}$. In addition, define the operator $\hat{B}=\sum_{j}^n \hat{X}_j$ where $\hat{X}_j$ is the Pauli $X$ operator acting on qubit $j$. Given a set of parameters $\bs{\gamma}=(\gamma_1,\dots,\gamma_p)\in \mathbb{R}^p$ and $\bs{\beta}=(\beta_1,\dots,\beta_p)\in \mathbb{R}^p$, the QAOA prepares the initial state as $\ket{s}=\ket{+}^{n}=2^{-n/2}\sum_{\bs{z}}\ket{\bs{z}}$ and applies $p$ layers of alternating unitary operators $e^{-i\gamma_k \hat{C}}$ and $e^{-i\beta_k \hat{B}}$ to prepare the state
\begin{align}
    \ket{\bs{\gamma,\beta}}=e^{-i\beta_p \hat{B}}e^{-i\gamma_p \hat{C}} \dots
    e^{-i\beta_1 \hat{B}}e^{-i\gamma_1 \hat{C}} \ket{s}.
\end{align}
For a given cost function $C$, the corresponding QAOA objective function is the expectation value $\expval{\hat{C}}{\bs{\gamma,\beta}}$. The goal in the QAOA at depth $p$ is to find the $2p$ optimal parameters $(\bs{\gamma}^*,\bs{\beta}^*)$ that maximises the expectation value $\expval{C}{\bs{\gamma,\beta}}$ (i.e.\ $\max_{\bs{\gamma},\beta} \expval{C}{\bs{\gamma,\beta}}$).
Heuristics strategies to optimize $\expval{C}{\bs{\gamma,\beta}}$ with respect to $(\bs{\gamma,\beta})$ using a good initial guess have been proposed in \cite{heuristics_QAOA}. Recently, when applying the QAOA to large girth regular graphs with girth greater than $2p+1$, the graph appears as a regular tree and the authors of \cite{QAOAspinglass} developed an algorithm to find the optimal parameters which we call tree-parameters.  

\section{Summary of known theorems}
\label{sec:results}

\subsection{QAOA and COPs}
The first result of the QAOA on spin glass can be found in \cite{SKQAOAFarhi2022quantumapproximate} where the authors applied the QAOA on the Sherrington--Kirkpatrick (SK) model and found an algorithm to evaluate the expectation value in the infinite limit after averaging over the disorder $\mathbb{E}_J$. More generally, for a $q$-spin glass with cost function
\begin{align}
    H_q(\bs{z}) = \sum_{j<k<\dots <q} J_{jk\dots q} z_j z_k\dots z_q,
\end{align}
it was shown in \cite{Basso_2022} that the following theorem holds

\begin{theorem}[Theorem 1 of \cite{Basso_2022}]
    \label{thm:Vqspinglass}
    For any $p$ and any parameters $(\bs{\gamma,\beta})$, we have
    \begin{align}
        \lim_{n\rightarrow \infty} \mathbb{E}_J
        \left[ \expval{H_q/n}{\bs{\gamma,\beta}} \right]
        =V^{(q)}_p(\bs{\gamma,\beta}),
    \end{align}
\end{theorem}
In \cite{QAOAspinglass}, the authors evaluated the performance of the QAOA for Max-q-XORSAT  on large-girth $(D+1)$-regular graphs. By restricting to graphs that are regular and girth (also known as the shortest Berge-cycle) greater than $2p+1$, the subgraph explored by the QAOA at depth $p$ will appear as regular trees. Since the optimal cut fraction is of the form $1/2 + \mathcal{O}(1/\sqrt{D})$ in a typical random graph as in \cref{eq:asym_xorsat}, we have
\begin{align}
    \label{eq:cut_Fraction}
    \frac{1}{|E|}\expval{ H_{XOR}^q}{\boldsymbol{\gamma},\boldsymbol{\beta}}=\frac{1}{2}+\nu_p^{[q]}(D,\boldsymbol{\gamma},\boldsymbol{\beta})
        \sqrt{\frac{q}{2D}} + \mathcal{O}(1/\sqrt{D}).
\end{align}
Let 
\begin{align}
    \nu_p^{[q]} (\bs{\gamma,\beta}) = \lim_{D\rightarrow \infty }
    \nu_p^{[q]}(D,\boldsymbol{\gamma},\boldsymbol{\beta}),
\end{align}
then, we have the following theorem

\begin{theorem}[Theorem 2 of \cite{QAOAspinglass}]
    \label{thm:CXOR_algorithm}
    For $H_{XOR}^q$ on any $(D+1)$-regular $q$-uniform hypergraphs with girth $>2p+1$, for any choice of $J$, \cref{eq:cut_Fraction} can be evaluated using $\mathcal{O}(p 4^{pq})$ time and $\mathcal{O}(4^{p})$ space.    
    In addition, the infinite $D$ limit can be evaluated with an iteration using $\mathcal{O}(p^2 4^p)$ time and $\mathcal{O}(p^2)$ space.
\end{theorem}
Furthermore, the authors also made the following conjecture based on promising numerical evidence,
\begin{conjecture}[Conjecture of \cite{QAOAspinglass}]
    Optimising the QAOA using tree-parameters found in \cite{QAOAspinglass}, the Parisi value for the Sherrington--Kirkpatrick model can be reached:
    \begin{align}
        \lim_{p\rightarrow \infty} \nu_p^{[2]} (\bs{\gamma,\beta}) = \Pi_2.
    \end{align}
    \begin{remark}
        Note that to compute $\nu_{p+1}^{[2]} (\bs{\gamma,\beta})$, one technically first computes $\nu_{p+1}^{[2]} (D,\bs{\gamma,\beta})$ before taking the large $D$ limit $ \lim_{D\rightarrow \infty} \nu_{p+1}^{[2]} (D,\bs{\gamma,\beta})$ so the statement is saying the $\Pi_2 =\lim_{D\rightarrow \infty}\lim_{p\rightarrow\infty} \nu_{p}^{[2]} (D,\bs{\gamma,\beta}) $.
    \end{remark}
\end{conjecture}
While we are unable to prove or disprove the conjecture, we prove later that the generalised Parisi value $\Pi_q$ is not obtainable if the OGP is present.\\

One point to note is that in \cite{QAOA_seegraph}, it has been shown that at low depth $p$, if a problem exhibits the OGP, then the locality of the QAOA makes it such that it is prevented from getting close to the optimal value if it does not see the whole graph. Specifically, the following theorem is proven
\begin{theorem}[modified version of Corollary 4.4 in \cite{chou2022limitations}]
    \label{thm:algo_CLS}
    For Max-$q$-XORSAT on a random Erd\"os--R\'eyi directed multi-hypergraph, for every even $q\ge 4$, there exists a value $\eta_{OGP}< \eta_{OPT}$, where $\eta_{OPT}$ is the energy of the optimal solution, and a sequence $\{\delta(d)\}_{d\ge 1}$ with the following property: for every $\epsilon >0$ there exists sufficiently large $d_0$ such that for every $d>d_0$, every $p\le \delta(d)\log n$ and an arbitrary choice of parameters $\bs{\gamma,\beta}$ with probability converging to 1 as $n\rightarrow \infty$, the performance of the QAOA with depth $p$ satisfies $\expval{C_{XOR}^q/n}{\bs{\gamma,\beta}} \le \eta_{OGP} + \epsilon$.
\end{theorem}

The authors of \cite{QAOAspinglass} noted that assuming the OGP also holds for regular hypergraphs, then a similar argument can be used to show that the QAOA's performance as measured by the algorithm in \cref{thm:CXOR_algorithm} does not converge to the Parisi value $\Pi_q$ for even $q\ge 4$. This is because the large girth assumption implies that the graph has at least $D^p$ vertices so $p$ is always less than $\epsilon \log n$ in this limit. For the Max-$q$-XORSAT, the subgraph explored at constant $p$ has $q[(q-1)^pD^p +\dots +(q-1)D+1]$ vertices.  This lays the foundation of \cref{conjec} later.

\subsection{Equivalence of performance}
The first equivalence between spin glass and MaxCut for the QAOA was shown in \cite{QAOAspinglass}, where the performance of the QAOA at depth $p$ on the SK model as $n\rightarrow \infty$ is equal to the performance of the QAOA at depth $p$ on MaxCut problems on large-girth $(D+1)$-regular graphs when $D\rightarrow \infty$. In the follow up work of \cite{Basso_2022}, they generalize this result to show that the QAOA's performance for the $q$-spin model is equivalent to that for Max-$q$-XORSAT on any large girth $D$-regular hypergraphs in the limit $D\rightarrow \infty$.
\begin{theorem}[Theorem 3 of \cite{Basso_2022}]
    \label{thm:equivalence}
    Let $ \nu_p^{[q]}(\boldsymbol{\gamma},\boldsymbol{\beta})$ be the performance of the QAOA on any instance of Max-$q$-XORSAT that has an underlying $D$-regular, $q$-uniform hypergraph with girth $> 2p+1$ as given in \cite{QAOAspinglass}. Then for any $p$ and any parameters $(\bs{\gamma,\beta})$, we have
    \begin{align}
        \label{eq:equivalence}
        V^{(q)}_p(\bs{\gamma,\beta})=\sqrt{2}\nu_p^{[q]}(\sqrt{q}\boldsymbol{\gamma},\boldsymbol{\beta})
    \end{align}
\end{theorem}

The equivalence of performance of the QAOA on dense and sparse graph is also shown to hold in the case of  Erd\"os--R\'enyi graph. 

\begin{theorem}[modified version of theorem 2 in \cite{Basso_2022}]
    Let
    \begin{align}
        V_p(\mathbb{G},\gamma,\beta) = \lim_{n\rightarrow\infty} \mathbb{E}_{J\sim \mathbb{G}(n)}
        \expval{H_J/n}{\bs{\gamma,\beta}},
    \end{align}
    where $\mathbb{G}$ denotes the underlying graph and $H_J$ the cost function associated with it. Then, for the $q$-spin model $\mathbb{G}_q$ and the  Erd\"os--R\'enyi graph with connectivity $\lambda$, the asymptotic performance of the QAOA on $\mathbb{G}_q$ is the same as $\mathbb{G}^q_{ER}$ for any $(\gamma,\beta)$
    \begin{align}
         V_p^{(q)}(\gamma,\beta) = \lim_{\lambda \rightarrow \infty}  V_p(\mathbb{G}^q_{ER},\gamma,\beta).
    \end{align}
\end{theorem}

\section{Main Results}
\label{sec:Conjecture}

We now have the pieces in place to state our main theorem
\begin{theorem}
    \label{thm:algo_limit}
    Given a local algorithm $\mathcal{A}$ that is limited in performance up to depth $p= \epsilon \log n$ on an Erd\"os--R\'enyi hypergraph with sufficiently large average degree $\lambda$, $\mathcal{A}$ is also limited in performance up to depth $p$ on a random $D$-regular hypergraph for sufficiently large $D$.
\end{theorem}

We delay the proof for \cref{subsec:proof} and note an immediate consequence for the performance of the QAOA.

\begin{theorem}
    \label{conjec}
    For Max-$q$-XORSAT on a D-regular $q$-uniform hypergraph, for every even $q\ge 4$, there exists a value $\eta_{OGP}$ such that it is smaller than the optimal value $\eta_{OPT}$ with the following property. For every $\epsilon >0$ there exists sufficiently large $d_0$ such that for every $d>d_0$, every $p\le d\log n$ and an arbitrary choice of parameters $\bs{\gamma,\beta}$ with probability converging to 1 as $n\rightarrow \infty$, the performance of the QAOA with depth $p$ satisfies
    \begin{align}
        \frac{\expval{C_{XOR}^q/n}{\bs{\gamma,\beta}}}{|E|} = 
        \frac{1}{2}+
    \nu_{\infty}^{[q]}(D,\boldsymbol{\gamma},\boldsymbol{\beta})
        \sqrt{\frac{q}{2D}} + \mathcal{O}(1/\sqrt{D})<\eta_{OGP} + \epsilon.
    \end{align}
\end{theorem}
\begin{proof}
    The proof of the inequality follows from \cref{thm:algo_CLS} and \cref{thm:algo_limit}.\\

    The proof the the equality comes from \cref{thm:algo_limit} which states that evaluating $\nu_{p}^{[q]}(D,\boldsymbol{\gamma},\boldsymbol{\beta})$ is valid up to $p=\epsilon \log n$. In the case of MaxCUT or Max-$q$-XORSAT, we have
\begin{align}
    \lim_{n\rightarrow\infty}
    \lim_{p\rightarrow \epsilon \log n}
    \frac{1}{|E|}\expval{ H_{XOR}^q}{\boldsymbol{\gamma},\boldsymbol{\beta}}
    &=\frac{1}{2}+
    \nu_{\infty}^{[q]}(D,\boldsymbol{\gamma},\boldsymbol{\beta})
        \sqrt{\frac{q}{2D}} + \mathcal{O}(1/\sqrt{D})\\
    &<     \frac{1}{2}+
    \Pi_q
        \sqrt{\frac{q}{2D}} + \mathcal{O}(1/\sqrt{D}),
\end{align}
where $\nu_{\infty}^{[q]}(D,\boldsymbol{\gamma},\boldsymbol{\beta})= \lim_{p\rightarrow \infty}\nu_{p}^{[q]}(D,\boldsymbol{\gamma},\boldsymbol{\beta})<\nu_{\infty}^{[q]}(\boldsymbol{\gamma},\boldsymbol{\beta})<\Pi_q$.
\end{proof}

As a result of this theorem, we re-derive and improve a theorem of \cite{Basso_2022}:
\begin{corollary}[Theorem 4 of \cite{Basso_2022}]
    \label{thm:novel}
    From \cref{conjec} the performance of the QAOA on the pure $q$-spin glass for even $q\ge 4$ is upper bounded by $\eta_{OGP}$ as $p\rightarrow \infty$ and is strictly less than the optimal value, i.e.\ the Parisi value $\Pi_q$, under the swapping of limits.
\end{corollary}
We emphasise again that while \cite{Basso_2022} proves this corollary, it does so via explicit calculation of the QAOA using a highly technical proof involving a generalized multinomial sum. This work improves upon it by providing a simpler, straightforward argument using graph properties and extending their result from constant depth to $\epsilon \log n$ depth for the random regular graph. Note that both proofs require \cref{thm:algo_CLS} to prove this corollary.\\

Unfortunately, the QAOA performance on mean-field spin glass has only been shown to be equivalent in performance to the dilute spin glass model via \cref{eq:equivalence} at constant depth $p$ and we are unable to find a proof to extend their equivalence. As such, we are not able to prove that the OGP acts as a barrier on dense graph at logarithmic depth. If one can show that an asymptotic analysis of the mean-field spin glass at constant depth and logarithmic depth result in the same solution as in the random regular graph, this could extend the algorithmic barrier for dense graphs to logarithmic depth as well.\\

We note that \cref{thm:novel} shows that the QAOA will not be able to find the optimal value even when it sees the whole graph and the algorithm runs indefinitely if one optimises the parameters of the QAOA using the tree-parameters obtained via asymptotic analysis. Formally, the Parisi value is attainable via the QAOA with the following limits
\begin{align}
    \lim_{n\rightarrow \infty}\lim_{p\rightarrow \infty} 
    \mathbb{E}_J
    \expval{H_q/n}{\bs{\gamma},\bs{\beta}}=\Pi_q.
\end{align}
 The iteration provided in \cref{thm:CXOR_algorithm} swaps the limits which results in failure of the QAOA to find the optimal value. This leads us to the following corollaries
\begin{corollary}
    \label{cor:spinglass}
    If the OGP exists in a spin glass type problem, then the swapping of limits results in a sub-optimal solution for both random regular graphs and Erd\"os--R\'enyi graphs. In other words, a necessary condition for the validity of limit swapping is that the problem does not exhibit the OGP.
\end{corollary}

\begin{remark}
    For $q$-spin glass, it is expected that OGP holds for all $q\ge 3$ which suggests that limit swapping is not allowed for all mean-field spin glasses with the possible exception for the $2$-spin glass model (i.e.\ the SK model) \cite{SK_infRSB}.
\end{remark}

In addition, \cref{cor:spinglass} also applies to the Maximum Independent Set problem since \cite{QAOA_seegraph} shows a similar limitation at logarithmic depth. We suspect that similar results extend to all COPs rather than being limited to spin glass type COPs.

\subsection{Proof of \cref{thm:algo_limit}}
\label{subsec:proof}
\begin{proof}

The proof is as follows: first, we need to show that for some $\lambda\in \mathbb{Z}^+$, a $\lambda$-regular $q$-uniform hypertree can be embedded into an Erd\"os--R\'enyi hypergraph of sufficiently high average connectivity. Then, we show that a $\lambda$-regular $q$-uniform hypergraph can be generated from said hypertree. Finally, we show that algorithm $\mathcal{A}$ must also fail to find solutions arbitrarily close to the optimal solution in a $\lambda$-regular $q$-uniform hypergraph as doing so would result in a contradiction.\\

We note that in a COP instance with $m$ chosen edges can be converted into a regular instance changing only $o_{\lambda}(1/\sqrt{\lambda})$\footnote{By $o_{\lambda}(g(\lambda))$, we mean that there is a function $f(n,\lambda)$ such that $\frac{f(n,\lambda)}{g(\lambda)}\rightarrow 0$ as $n\rightarrow \infty$ for fixed $\lambda$, then setting $\lambda\rightarrow \infty$.} edges. This reduction has been used several times before as in \cite{Dembo_2017,sen2017optimization} and recently in \cite{chen2023local}. We follow the proof of \cite{chen2023local} but simplify one of the lemma using $G_{ER}^q(n,p)$ later to avoid the the combinatoric arguments needed when considering $G_{ER}^q(n,m)$, and extend the argument from constant depth $l$ to $l \le \epsilon \log n$.\\

Here, we show that an Erd\"os--R\'enyi hypergraph can be converted into a hypertree by changing only $o(1/\sqrt{\lambda})$ edges. Given $\mathbb{G}^q_{ER}(n,m)$ with average degree $\lambda$, define $\lambda'= \lceil \lambda + \sqrt{\lambda} \log \lambda \rceil$. Let $d_i$ be the degree of vertex $v_i$. Modify the graph as follows:
\begin{enumerate}
    \item Remove edges until $d_i \le \lambda'$ for all vertices.
    \item Add edges to all vertices until $d_i=\lambda'$ where each vertex is chosen with probability proportional to $\lambda'-d_i$.
\end{enumerate}

In order to prove this, we will need the following lemmas.

\begin{lemma}
    In the limit $n\rightarrow \infty$, the number of edges removed  from $\mathbb{G}^q_{ER}(n,m)$ is at most $n\cdot \mathcal{O}_{\lambda}(1/\lambda^{c\log \lambda})$ for some constant $c>0$.
\end{lemma}

\begin{proof}
    Note that the distribution of edges in an Erd\"os--R\'enyi follows a binomial distribution $B(m,1/n)$. For each vertex $v_i$, the number of edges removed is either 0 or $d_i-\lambda'$ so $\Delta_i:=\max$ ($\lambda'$,0). The first moment of $\Delta_i$ is bounded by
    \begin{align}
        \mathbb{E}[\Delta_i] = \sum_{d>\lambda'}P(\Delta_i=k)(k-\lambda')
        =  \sum_{d>\lambda'}P(\Delta_i\ge k) \le \int_{\lambda'}^{\infty} P (d_i \ge x) dx.
    \end{align}
    Where we used the fact that the expectation value of a random variable equals the cumulative function in the second equality.\\

    The second moment can also be bounded  as 
    \begin{align}
         \mathbb{E}[\Delta_i^2] = \sum_{d>\lambda'}P(\Delta_i=k)(k-\lambda')^2
        \le \int_{\lambda'}^{\infty} P (d_i \ge x) 2(x-\lambda')dx.
    \end{align}
    Using Chernoff's bound for the binomial distribution, we have
    \begin{equation}
        \label{eq:chernoff}
        P (d_i \ge x) \le 2 \exp{-\frac{(x-\lambda')^2}{3\lambda'}}.
    \end{equation}
    Applying \cref{eq:chernoff} to the second moment bound, we have
    \begin{align}
        \int_{\lambda'}^{\infty}2(x-\lambda') \exp{-\frac{(x-\lambda')^2}{3\lambda'}}dx &= 3\lambda' \exp{-\frac{(\lambda'-\lambda)^2}{3\lambda}} \cr
        &= \mathcal{O}_{\lambda}\left(\lambda \exp{-c (\log \lambda )^2}
        \right).
    \end{align}
    for some $c>0$. Thus, both the first and second moment are bounded by $\mathcal{O}_{\lambda} (d^{-c' \log d})$ for some constant $c'>0$.\\

    For the total number of edges removed $\Delta$, we note that by a union bound, $\Delta \le \Tilde{\Delta}$, where $\Tilde{\Delta}:= \sum_i \Delta_i$. Furthermore, we have $\mathbb{E}[\Tilde{\Delta}]=n \mathbb{E}[\Delta_i]$. Unless $\Delta_i$ and $\Delta_j$ share the same edge, they are independent, so 
    \begin{align}
        \mathrm{Var}(\Tilde{\Delta})\le n \mathbb{E}[\Delta_i^2]+ 2 \sum_{i,j} \mathrm{Cov}(\Delta_i,\Delta_j).
    \end{align}
    Since the degree of each vertex is not independent (i.e.\ follows a multinomial distribution), the covariance term is negative. Therefore, as $n\rightarrow \infty$, $\Delta$ is at most $n \mathcal{O}_{\lambda}(d^{-c\log d})$ for some $c>0$ with high probability.
\end{proof}

The process of removing edges does not create cycles (i.e.\ destroy tree-like property). However, we need to ensure that the graph was initially tree-like and remains tree-like after adding edges. Rather than working with $G=\mathbb{G}^q_{ER}(n,m)$, we will work with $G=\mathbb{G}^q_{ER}(n,p)$ for the next lemma and use the fact that a graph containing a $k$-cycle is a monotone increasing property and that for any monotone increasing graph property $\mathcal{P}$, $P(\mathcal{P}\in \mathbb{G}^q_{ER}(n,m))\le C \cdot P(\mathcal{P}\in \mathbb{G}^q_{ER}(n,p))$ for some constant $C$ \cite{frieze}.

\begin{lemma}
    Fix any constant $\lambda$ and $l \le \epsilon \log n$. With high probability as $n\rightarrow \infty$, 1-$o(1)$ fraction of the $l$-local neighbourhood are treelike.
\end{lemma}

\begin{proof}
    Let $p= \frac{c+\log n + \lambda \log \log n}{{{n}\choose {q-1}}} \sim \mathcal{O}(\log{n}/n^{q-1})$ for some constant $c>1$, and $k\in \mathbb{N}$. Let $X$ be the number of $k$-cycles. Leaving $k$ in the big-$\mathcal{O}$ notation to account for $k$ as a function of $n$ later, we have
    \begin{align}
        \mathbb{E}X = {n \choose k} p^k \sim \mathcal{O}\left((q!)^k  \frac{n^{2k}\log^k{n} }{n^{kq}k!}\right).
    \end{align}
    By Markov's inequality, we thus have
    \begin{align}
        P(X>0) \le \mathcal{O}\left( \frac{(q!)^k\log^k{n}}{n^{k(q-2)} k!}\right),
    \end{align}
    which vanishes in the limit $n\rightarrow \infty$ implying that the number of constant $k$-cycles is $o(1)$. The same argument implies that cycles of size $\log n$ and $(\log{n})^{\log n}$ also vanish.\\

    Now we show that the $l$-local neighbourhood of an arbitrary vertex has at most $o(1)$ $k$-cycles. In the limit $n\rightarrow \infty$, the degree of each vertex follows a Poisson distribution  with mean $\lambda$. Let $Y$ denote the degree of an arbitrary vertex. The probability that any vertex $v$ has degree at most $\log n$ is given by
    \begin{align}
        P(Y \le \log n) = 1 - P(Y > \log n) 
    \end{align}
    From the Chernoff bound for the Poisson distribution, we have that
    \begin{align}
        P(Y>x)\le e^{-\lambda} \frac{(e\lambda)^x}{x^x} \sim \mathcal{O}(c^x/x^x)= o(1),
    \end{align}
    for some constant $c$ so all vertices have degree at most $\log n$ with high probability.\\
    
    Thus, the $l$-local neighbourhood has at most $(\log n)^{\epsilon \log n}$ vertices. Repeating the same argument for the number of $k$-cycles shows that only $o(1)$ of the $l$-local neighbourhood will contain a cycle.
\end{proof}

\begin{lemma}
    Fix any $\lambda$, there exists some $\epsilon >0$ such that for $l\le \epsilon \log n$, with high probability as $n\rightarrow \infty$, adding edges preserves trees in 1-$o(1)$ fraction of the $l$-local neighbourhood.
\end{lemma}

\begin{proof}
    Right after removing edges, every vertex has at most degree $\lambda'$ so given some constant $\lambda$ and $l\le \epsilon \log n$, the $l$-local neighbourhood $B_G(v,l)$ is upper-bound by $\lambda^{\prime \epsilon \log n}$ and is a hypertree. Then, we have to add on average $n(\lambda \log{\lambda})\sim \Theta(n)$ edges but since $B_G(v,l)$ is of order $\mathcal{O}(\lambda^{\prime \epsilon\log n})$, the probability that an added hyperedge contains at least two vertices in $B_G(v,l)$ is $\mathcal{O}(\lambda^{\prime \epsilon \log n}/n)$.\\

    Choose $\epsilon < 1/ \log \lambda'$. As a result, adding clauses results in $o(1)$ fraction of $l$-local neighbourhood forming a cycle.
\end{proof}
Now, we show that a $\lambda$-regular, $q$-uniform hypergraph is locally also a hypertree.
\begin{lemma}
    Fix any $\lambda>1$ and $p\le \epsilon \log n$ for some $\epsilon >0$, with high probability as $n\rightarrow \infty$, $1-o_{\lambda}(1)$ fraction of vertices in the $p$-local neighbourhood are treelike.
\end{lemma}
\begin{proof}
    As we are interested in the large $n$ limit, we first show that for fixed $p$, the dominant term in the probability that a cycle is form in the large $n$ limit is given by $\frac{(q-1)^p \lambda^p} {n-1-\dots-\lambda^{p-1}(q-1)^{p-1}}$. Consider $p=1$ and choose any hyperedge. Then, the first $(q-1)$ vertices form no cycle with probability 1. The next hyperedge added will form a cycle with probability $\frac{q-1}{n-1}$. This process repeats until we reach the last hyperedge for the root vertex (i.e.\ the $\lambda$ hyperedge) where the probability of forming a cycle is given by $\frac{(\lambda -1)(q-1)}{n-1} $ so the dominant term is of the form $\frac{\lambda (q-1)}{n-1}$. In other words, the term that contributes the highest probability of forming a cycle at depth $p$ is when we are filling up the last hyperedge.\\ 
    
    For $p=2$ and higher, choosing the first hyperedge already has a non-trivial probability of forming a cycle as we might add a vertex at the $p-1$ level. Focusing on $p=2$ this means that adding the first $(q-1)$ vertices has a probability of $\frac{\lambda (q-1)-1}{n-1}$ to form a vertex. If we are in the middle of filling up the second layer (i.e.\ some of the $p=1$ vertices already have degree $\lambda$), then the adding the next hyperedge and vertex would form a cycle with a $p=1$ vertex with $\frac{\lambda (q-1) -c }{n-1-c}$ for some constant $c$ while the probability that it forms a cycle with a $p=2$ vertex is given by $\frac{c*\lambda (q-1)}{n-1-\lambda(q-1)}$. For the very last hyperedge added in $p=2$, the probability of forming a cycle is given by $\frac{(q-1)^2 (\lambda^2-1)}{n-1-\lambda(q-1)}$ which is the dominant term. This process can be iterated to show that the dominant term is of the form $\mathcal{O}(\frac{q^p \lambda^p}{n-\dots -\lambda^{p-1}(q-1)^{p-1}})$ at depth $p$.\\
    
    For any fix $\lambda$ and $p \le \epsilon \log n$ for some constant $\epsilon >0$, the probability that at $p$ distance away from any vertex $v_i$ remains tree-like is given by
    \begin{align}
    \max \left(0, 1-\frac{(q-1)\lambda}{n-1}-\dots - \frac{(q-1)^p \lambda^p}{n-1-\dots -\lambda^{p-1}(q-1)^{p-1}}
    \right),
    \end{align}
    
    since
    \begin{align}
        \mathcal{O}\left( \frac{c^{\epsilon \log n}}{n}\right)
        =
        \mathcal{O}\left( \frac{n^{\epsilon \log c}}{n} \right),
    \end{align}
    choose $\epsilon< \frac{1}{\log c}= \frac{1}{\log (\lambda q)}$ so that $\lim_{n\rightarrow \infty}(n^{\epsilon \log c -1})$ goes to 0 . Then the probability that at $\epsilon \log n$ distance away from any vertex is tree like converges to unity for $n\rightarrow \infty$,
    
    \begin{align}
        \label{eq:rrprob}
        \lim_{n\rightarrow \infty}1-\dots -\frac{(q-1)^p \lambda^p}{n-1-\dots -\lambda^{p-1}(q-1)^{p-1}} = 1.
    \end{align}
    
\end{proof}

Now we can show that the OGP is also an obstruction in random regular hypergraph via contradiction. Assuming that an algorithm $\mathcal{A}$ at logarithmic depth is able to find solutions arbitrarily close to the optimal solution for the Max-$q$-XORSAT on a regular hypergraph. Then this would imply that $\mathcal{A}$ is also able to find such solutions when performed on an Erd\"os--R\'enyi hypergraph since both graphs are $p$-locally the same. However, this contradicts \cref{thm:algo_CLS} and thus, the OGP must also restrict the performance of logarithmic depth local algorithms when applied to a regular hypergraph.
\end{proof}

It is of note that proving that the OGP exists in a problem is much easier when the underlying graph is an Erd\"os--R\'enyi hypergraph as compared to a regular hypergraph since only the former can be described by a probability distribution. This is why there is no proof that the OGP exists for the Max-$q$-XORSAT on regular graph as it requires the Poisson distribution found in an Erd\"os--R\'enyi hypergraph. Given \cref{thm:algo_limit} and that it is possible to show that the OGP exists in both Erd\"os--R\'enyi hypergraph and regular hypergraph in some problems \cite{LLAoversparse}, it is reasonable to think that if the OGP exists in the former, it also exists in the latter. Motivated by this, we make the following conjectures

\begin{conjecture}
    If the overlap gap property exists in a COP with an underlying Erd\"os--R\'enyi hypergraph of sufficiently high connectivity, then it also exists when the underlying hypergraph is a regular hypergraph of sufficiently high degree.
\end{conjecture}
\begin{remark}
    While \cref{thm:algo_limit} seems to support this conjecture, we have only proved this asymptotically and at logarithmic depth. It is not entirely clear if the same result applies when the depth is of $\mathcal{O}(n)$ or for the finite case since the proof of OGP for max-$q$-XORSAT on an Erd\"os--R\'enyi holds for some constant $n\in \mathbb{Z}$.
\end{remark}

\begin{conjecture}[Monotonicity of the OGP]
    For the Max-$q$-XORSAT problem, the overlap gap property is a monotonically increasing graph property.
\end{conjecture}
We note that the proof of \cref{thm:algo_limit} is much simpler if the conjectures are true as can be seen in \cref{app:proof}.

\subsection{Numerical evidence}

 \begin{figure}
        \centering
        \includegraphics[width=0.5\linewidth]{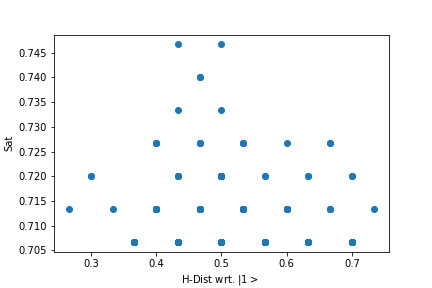}
        \caption{Plot of a typical cluster of solutions. Y-axis indicate fraction of clauses satisfied and X-axis represent the Hamming distance from the state $\ket{1}^{\otimes n}$ (e.g.\ the state $\ket{-1}^{\otimes n/2}\otimes \ket{1}^{\otimes n/2}$ has a value 0.5 on the X-axis.}
        \label{fig:enter-label}
    \end{figure}
    \begin{figure}[h]
        \centering
        \begin{subfigure}{0.475\textwidth}
            \includegraphics[width=\linewidth]{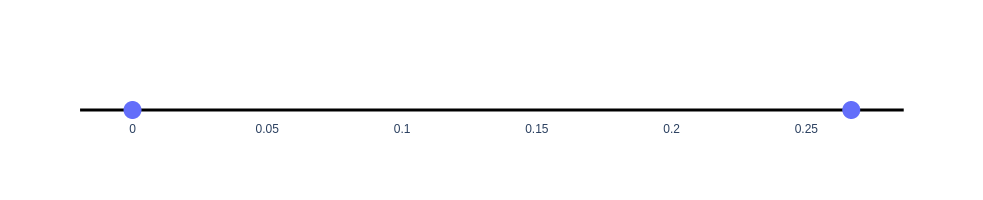}
            \caption{Overlap of optimal solutions.}
            \label{fig:optsol}
        \end{subfigure}
        \hfill
        \begin{subfigure}{0.475\textwidth}
            \includegraphics[width=\linewidth]{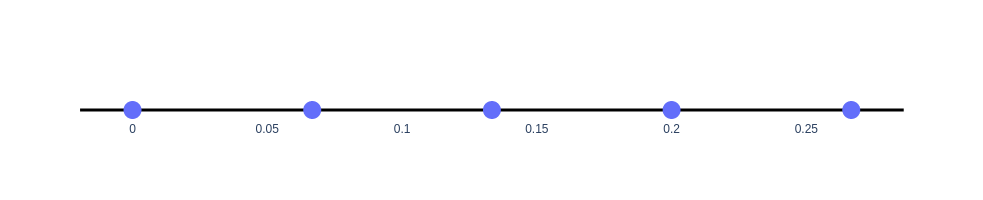}
            \caption{Overlap including first sub-optimal solution.}
            \label{fig:subsol}
        \end{subfigure}
        \vskip\baselineskip
        \begin{subfigure}{0.475\textwidth}
            \includegraphics[width=\linewidth]{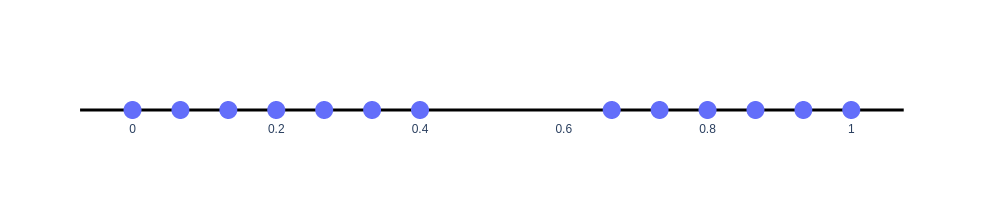}
            \caption{Overlap including first 2 sub-optimal solution.}
            \label{fig:sssol}
        \end{subfigure}
        \hfill
        \begin{subfigure}{0.475\textwidth}
            \includegraphics[width=\linewidth]{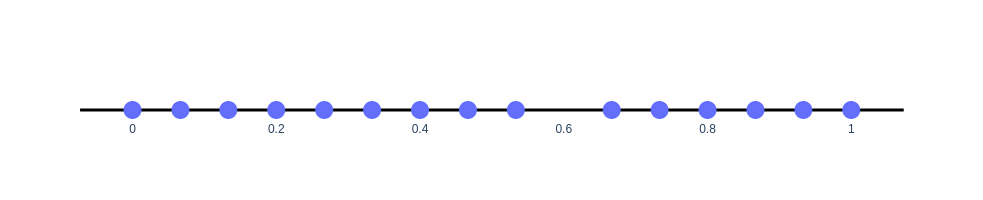}
            \caption{Overlap including first 3 sub-optimal solution.}
            \label{fig:ssssol}
        \end{subfigure}
        \caption{Typical evolution of the overlap spectrum of an $n=30$ Max-3-XORSAT instance with constant degree 15. Here, plot is based on the Hamming distance rather than overlaps. Typically, the optimal solution and the first few optimal solution does not exhibit the OGP. At some distance $\epsilon$ away from the optimal solution, the OGP occurs. Increasing $\epsilon$ further includes additional sub-optimal solution until the set of overlaps is dense.}
    \end{figure}

We refer the reader to numerous numerical studies about how the performance of the QAOA is unable to surpass the OGP barrier. For instance, the authors of \cite{QAOAspinglass} numerically evaluated the performance of the QAOA on Max-3-XORSAT up to $p=14$ and got $0.6623 \Pi_3$. For the $3$-spin glass, the OGP inhibits the AMS algorithm performance to get to $0.987\Pi_3$ \cite{alaoui2020algorithmic}. A study of the QAOA on Max-$q$-XORSAT problem similarly finds that for $n=18$, the QAOA is unable to get close to the 0.987 approximation ratio even at a depth of $p=30$ for $q=3$ \cite{QAOA_XORSAT}.\\

Instead, we provide some numerical evidence that instances of the OGP can occur in random regular hypergraph of odd degree. The code can be found here \cite{OGPCodes}. Our numerical simulation proceeds in the following manner. First, we define the problem size $n$, uniformity $q$, and degree $d$, where we implicitly assume that $nd$ is a multiple of $q$. Then, randomly generate a $d$-regular $q$-uniform hypergraph so that the total number of hyperdeges $|E|=(nd/q)$. Next, we randomly generate the list $\bs{J}=\{-1,+1\}^{|E|}$ for the coupling strength of the hyperedges. Finally, we perform a branch and bound algorithm and record those whose cut-fraction exceeds a certain threshold.\\

Once we have the list of bit-strings and their corresponding cut-fraction, we have to choose some $\epsilon >0$ such that the list of bit-strings that are $\epsilon$-optimal solutions is small. By default, we limit the bit-strings that are at least 95\% to the optimal solution. Finally, compute the overlap between all such $\epsilon$-optimal bit-string and obtain the overlap spectrum. \\

    We find that on average, when $d<q$, the OGP is not present. It is only when $d$ is greater than $q$ that instances of problems exhibiting the OGP first appears. The numerical simulations was run on $q=3$ and varying $n$ up till 30.\\

    We also ran simulations on the SK model as is it highly believed, though not yet proven, that the SK model does not exhibit the OGP \cite{Gamarnik_2022,alaoui2020algorithmic}. We find that indeed the SK model does not exhibit the OGP at $n=45$.

    \begin{figure}[t]
    \centering
    \includegraphics[width=\textwidth]{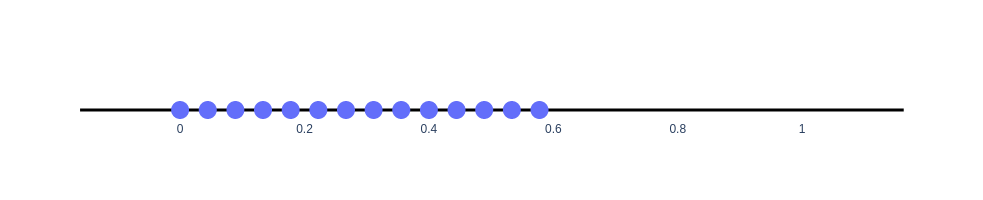}
    \caption{The overlap spectrum of an instance of the SK model with $=45$. No gap was observed and including additonal suboptimal solution merely monotonically increases the overlap until the set of overlap is dense.}
    \label{fig:noOGP}
\end{figure}

\section{Discussion and further work}
Being a heuristic algorithm, the limitations and potential of the QAOA have not yet been fully explored. While swapping the order of limits allows us to evaluate the expectation value with a classical computer faster, it also seems to lead to sub-optimal results. This of course is expected and one can instead use the algorithm developed in \cite{QAOAspinglass} as a heuristic starting ansatz for $(\bs{\gamma,\beta})$ to be further optimized for a specific problem.\\

Currently, the OGP has only been shown to a limitation on dense model at super-constant depth $p\sim \mathcal{O} (\log \log n)$ \cite{Anshu2023concentrationbounds}. Given the ``dense-from-sparse'' reduction performed in \cite{Basso_2022} that showed equivalence between constant depth QAOA for dense and sparse graphs, perhaps one can extend the equivalence to logarithmic depth similar to how we extended the validity of the algorithm from constant to logarithmic depth.\\

We note that this results suggests that at logarithmic depth, the performance of the QAOA equals that of AMS's algorithm for the mean field spin glass \cite{alaoui2020algorithmic}, a type of Approximate Message Passing (AMP) algorithm. This suggests that if the QAOA is optimized correctly beyond logarithmic depth such as polynomial depth, it should outperform such classical algorithms since the QAOA is known to find exact solutions by reduction to the Quantum Adiabatic Algorithm. It is still an open question to determine at what depth $p$ will the QAOA outperform the algorithm. Furthermore, given the similarity in performance to the AMP algorithm, this also suggests that the conjecture in \cite{QAOAspinglass} that the Parisi value for the Sherrington--Kirkpatrick model is obtained under limit swapping might be true as the AMP algorithm achieves the optimal value under the assumption that the OGP does not exists.

\section*{Acknowledgement}
We are thankful to Leo Zhou and Joao Basso for answering questions on their work that this paper is a follow up on, and David Gamarnik for valuable discussions about the overlap gap property. We thank David Gross, Matthias Sperl, and Michael Klatt for feedback on an early draft of this article. We also thank an anonymous referee for suggesting the main proof could be done via modifying \cite{chen2023local} replacing the original proof now in the appendix, and many others for corrections to improve the quality of the paper. This work was funded by DLR under the Quantum Computing Initiative.

	\clearpage

\bibliography{bibliography}
		\bibliographystyle{IEEEtran}

    \clearpage

\appendix
\section{Alternate proof}
\label{app:proof}
We provide here an alternate proof that relies on a conjecture that seems likely to be true as \cite{gamarnik2019landscape} notes that in the planted clique problem, the occurrence of the OGP is related to the monotonicity of another graph property. For the Max-$q$-XORSAT, it is reasonable to think that the OGP is related to the density of edges in the graph which is clearly monotonically increasing.
\begin{conjecture}
    The OGP of Max-$q$-XORSAT on any graph is a monotonic increasing property in the sense that if the graph $G$ exhibits the overlap gap property, then adding an additional edge does not destroy the graph exhibiting the OGP i.e.\ $G+e$ exhibits the OGP.
\end{conjecture}

    Another result that we need comes from the fact that we can embed an Erd\"os--R\'enyi hypergraph into a random regular hypergraph.
    \begin{theorem} [Theorem 1 and Corollary 2 of \cite{Dudek_2017}]
        For each $q\ge2$ there is a positive constant $C$ such that if for some real $\gamma=\gamma(n)$ and positive integer $d=d(n)$,
        \begin{align}
        \label{eq:gammacont}
        C\left(
        \left(d/n^{q-1}+\log(n)/d \right)^{1/3}
        +1/n
        \right)
        \le \gamma
        <1,
        \end{align}
        and $m=(1-\gamma)nd/q$ is an integer, then there is a joint distribution of $\mathbb{G}^{q}_{ER}(n,m)$ and $\mathbb{G}^{q}(n,d)$ with
        \begin{align}
            \lim_{n\rightarrow \infty} \mathbb{P}(\mathbb{G}^{q}_{ER}(n,m) \subset
        \mathbb{R}^{q}(n,d)) =1.
        \end{align}
        Furthermore, let $\mathcal{P}$ be a monotone increasing property. if $\log{n}\ll d \ll n^{q-1}$, for some $m\le (1-\gamma)nd/q$ where $\gamma$ satisfies \cref{eq:gammacont}, then if $\mathbb{G}^{q}_{ER}(n,m)\in \mathcal{P}$ as $n\rightarrow \infty$, then $\mathbb{R}^{q}(n,d)) \in \mathcal{P}$ as $n\rightarrow \infty$.
    \end{theorem}
    For finite $n$ and $d$, the proof trivially follows from the assumption of monotonicity. In the case for infinite $n$ and $d$, we have to prove that our condition for large potentially infinite girth fit this condition for the embedding.\\

    The maximum girth $g$ of a $d$-regular $q$-uniform hypergraph is known \cite{ellis2017regular} to be bounded by
    \begin{align}
      \frac{\log{n} -\log 4}{\log (q-1)+\log (d-1)}-1<  g \le \frac{2\log n}{\log (q-1)+\log (d-1)}+2.
    \end{align}
    For constant $q$, an infinite girth requires $d \ll n$\\
    
    Let $d\sim \mathcal{O}(n^{\epsilon})$ for sufficiently small $\epsilon>0$ and $n^{\epsilon}>q$. Such a constraint satisfies the large girth requirement. Substituting $d$ into \cref{eq:gammacont} gives us
    \begin{align}
        C\left(
        \left(\frac{n^{\epsilon}}{n^{q-1}}+\frac{\log{n}}{n^{\epsilon}}\right)^{1/3}
        + 1/n
        \right) \le \gamma <1,
    \end{align}
    where in the large $n$ limit, we see that the lhs.\ approaches 0. Thus, for $m=(1-\gamma)nd/q \sim \mathcal{O}(n^{1+\epsilon})$, an embedding can be performed.\\
    
    From this, it follows that for some $m^*\le (1-\gamma)nd/q$, if $\mathbb{G}^{q}_{ER}(n,m) $ has a monotone increasing property $\mathcal{P}$, then $\mathbb{R}^{q}(n,d)$ also has it as well.
    For $m$ to be drawn from a Poisson distribution where the graph has connectivity $\lambda$, we require $m=n (\log n + (\lambda-1)\log \log n +c)/d\sim \mathcal{O}(n^{1-\epsilon}\log n)$ for some finite constant $c$ \cite{poole2015strength}. Thus, there exists a $d_0$ such that, for $d>d_0$, the existence of the OGP is present in the solution space of Max-$q$-XORSAT on Random Regular hypergraph  $\mathbb{R}^{q}(n,d)$ with high probability meaning that the result of \cite{chou2022limitations} also applies to random regular hypergraphs.

\end{document}